\documentclass[preliminary,copyright,creativecommons]{eptcs}
\providecommand{\event}{TARK 2025} 

\usepackage{iftex}

\ifpdf
  \usepackage{underscore}         
  \usepackage[T1]{fontenc}
  \usepackage{caption}
  \usepackage{csquotes}        
\else
  \usepackage{breakurl}           
\fi

\usepackage{amsfonts}
\usepackage{mathtools}
\usepackage{amsmath}
\usepackage{amssymb}
\usepackage{tabularx}
\usepackage{hyperref}
\usepackage{booktabs}
\usepackage{multirow}
\usepackage{tikz}
\usepackage{pifont,dsfont} 
\usepackage{nicefrac}
\usepackage{todonotes}
\usepackage[shortlabels,inline]{enumitem}
\usepackage{marvosym}
\usepackage{caption}
\usepackage{subcaption}
\usepackage{bigdelim}

\usepackage[noend]{algpseudocode}
\usepackage{algorithm}

\usepackage{arydshln}

\usepackage{tcolorbox}

\usetikzlibrary{positioning,fit,calc}

\newcommand{\EP}[4]{
	\begin{center}
		\smallskip
		{\small 
			\begin{tabularx}{\columnwidth}{@{}l@{\hspace*{2mm}}l@{}}
				\toprule
				\multicolumn{2}{c}{\sc{#2}}\label{#1} \\
				\midrule
				{\bf Given:}& \parbox[t]{0.89\columnwidth}{#3\vspace*{1mm}} \\
				{\bf Question:}& \parbox[t]{0.89\columnwidth}{#4\vspace*{.5mm}} \\ 
				\bottomrule
			\end{tabularx}
		}
		\smallskip
	\end{center}
}

\newcommand{\N}{\mathbb{N}}

\newcommand{\littlep}{{\rm p}}
\newcommand{\manyone}{\ensuremath{\mbox{$\,\leq_{\rm m}^{{\littlep}}$\,}}}
\newcommand{\textmanyone}{\ensuremath{\mbox{$\leq_{\rm m}^{{\littlep}}$}}}

\newcommand{\p}{\ensuremath{\mathrm{P}}}
\newcommand{\np}{\ensuremath{\mathrm{NP}}}

\usepackage{pifont}
\newcommand{\cmark}{\ding{51}}%
\newcommand{\xmark}{\ding{55}}%

\newcommand{\OMIT}[1]{}

\DeclarePairedDelimiter\floor{\lfloor}{\rfloor}
\DeclarePairedDelimiter\abs{\lvert}{\rvert}
\DeclareMathOperator*{\argmax}{arg\,max}
\DeclareMathOperator*{\argmin}{arg\,min}

\newcommand\qedblob{\mbox{\ding{113}}}
\def\literalqed{{\ \nolinebreak\hfill\mbox{\qedblob\quad}}}

\newcommand{\eins}{\mathrm{1\hspace*{-0.9mm}I}}
\newcommand{\maj}{\mathrm{maj}}

\newcommand{\sumposi}[1]{\mathrm{sum\_pos}^{#1}}
\newcommand{\scorei}[1]{\mathrm{score}^{#1}}

\newcommand{\BucklinScorei}[1]{\mathrm{score}_{\mathrm{B}}^{#1}}

\newcolumntype{L}{>{\raggedright\arraybackslash}X}
\newcolumntype{R}{>{\raggedleft\arraybackslash}X}
\newcolumntype{C}{>{\centering\arraybackslash}X}

\title{Skating System Unveiled: Exploring Preference Aggregation in Ballroom Tournaments}
\author{Laryssa Horn\and
	Paul N\"{u}sken\and
	J\"{o}rg Rothe\and
	Tessa Seeger
	\institute{Institut f\"ur Informatik, MNF\\
Heinrich-Heine-Universit\"at D\"usseldorf\\
D\"usseldorf, Germany}
	\email{
	\{laryssa.horn, paul.nuesken, rothe, tessa.seeger\}@hhu.de
	}
}
\def\titlerunning{Exploring Preference Aggregation in Ballroom Tournaments}
\def\authorrunning{L. Horn et al.}

\usepackage{amsthm}
\theoremstyle{definition}
\newtheorem{definition}{Definition}
\theoremstyle{definition}

\theoremstyle{definition}

\theoremstyle{theorem}
\newtheorem{theorem}{Theorem}
\theoremstyle{example}
\newtheorem{example}{Example}
\theoremstyle{theorem}
\newtheorem{lemma}{Lemma}
\theoremstyle{theorem}

\theoremstyle{theorem}

\theoremstyle{theorem}

\begin{document}
\maketitle

\begin{abstract}
The \textit{Skating System}, which originated from the scrutineering system in dance sport tournaments, can be formulated as a voting system: We introduce and formalize the \emph{Skating System Single} (\emph{SkS}, for short), a new voting system embedded into the framework of computational social choice.
Although SkS has similarities with Bucklin voting, it differs from it because it is subject to additional constraints when determining the election winners.
Through an analysis of the axiomatic properties of SkS and of its
vulnerability to manipulative and electoral control attacks, we compare SkS with Bucklin voting and provide insights into its potential strengths and weaknesses.
In particular, we show that SkS satisfies nondictatorship as well as the majority criterion, positive responsiveness, monotonicity, and citizens' sovereignty but violates the Condorcet criterion, strong monotonicity, independence of clones, consistency, participation, resoluteness, and strategy-proofness.
Further, we study manipulation, i.e., where (groups of) voters vote strategically to improve the outcome of an election in their favor, showing that the constructive coalitional weighted manipulation problem for SkS is $\np$-complete, while the destructive variant 
can be solved in polynomial time.
Lastly, we initiate the study of electoral control, where an external agent attempts to change the election outcome by interfering with the structure of the election.
Here, we show \np-completeness for constructive and destructive control by deleting candidates as well as for constructive control by adding voters, whereas we show
that the problem of destructive control by adding voters can be solved in polynomial time.
\end{abstract}

\section{Introduction}
\label{sec:introduction}

Voting plays a fundamental role in collective decision-making and has an abundance of practical usages.
One of these applications is winner determination in \emph{ballroom tournaments}.
Ballroom dancing is a style of partner dance that originated in the 16th century in the European Royal courts~\cite{arbeau1888orchesographie,marion2008ballroom} and has since evolved, incorporating influences from various cultures.
It is characterized by specific partner dance positions, coordinated movements, and a combination of various dance styles.
Today, it continues to be popular worldwide, both as a social activity and a competitive sport with standardized rules and judging criteria.
A typical ballroom tournament proceeds in several rounds in which couples dance simultaneously and where the so-called adjudicators cast their votes on the participating couples.
After each round, at most half of the participating dancing couples
drop out, whereas the remaining couples proceed to the next round.
Eventually, the tournament progresses through the quarter-finals and semi-finals to the final.
A final consists of at least two and at most six couples competing in multiple dances.\footnote{In exceptional cases, the tournament office may allow more than six couples.}
The adjudicators evaluate the dance couples in two different ways depending on the phase of the tournament.
During all rounds except the final, the adjudicators approve or disapprove the participants.
However, they can give an approval to at most half of the participants.
The evaluation format changes in the finals, though, where the adjudicators are asked to cast a strict linear ranking of all participants per dance.
This ranking must be created independently for every dance, i.e., having, e.g., four dances in the final, all adjudicators must rank all couples four times.

The intricate and complex rule set the adjudicators must adhere to is known as the \emph{Skating System}, described in detail by Williams~\cite{williams2018skating}.
The system's name has its roots in figure skating and was adapted and later used for dance sport tournaments.
Records show that the first major competition that used an early form of the Skating System was the British Championship in Blackpool in 1937 (see \cite{williams2018skating}).
The rule set covers the following topics: Instructions to the adjudicators on how to evaluate the participants, winner determination of the single dances, winner determination of the tournament, and additional rules
to break possible ties.
That is, the Skating System consists of multiple evaluation methods, each producing an output to be used for the next evaluation step.
The following example shows
how to rank the couples---based on the rankings of the adjudicators---to determine the winner of one dance in a dance tournament final according to the Skating System.

\medskip
\noindent
\begin{minipage}[h!]{0.65\columnwidth}
	\begin{example}\label{example:dance-tournament}
		Consider a dance tournament final with a single dance, six competing couples, and five adjudicators---$A$, $B$, $C$, $D$, and $E$---each ranking the couples from first (best) to sixth (last) place as shown on the right.
		Intuitively, couple $31$ could be winning for having the most first positions, or couples $32$ or $34$ for having several of first, second, and third positions.
		However, the actual winner according to the Skating System is couple~$33$.
		Couple $32$ is ranked second, $34$ is ranked third, followed by couples~$36$, then~$35$, and lastly, couple $31$ is even in the last position.
	\end{example}
\end{minipage}
\hfill
\begin{minipage}[t]{0.32\columnwidth}
	\begin{tabular}{ c @{\quad} c c c c c }
	\toprule
	Couples & $A$ & $B$ & $C$ & $D$ & $E$ \\
	\midrule
	$31$ & $1$ & $1$ & $6$ & $6$ & $6$ \\
	$32$ & $4$ & $2$ & $3$ & $3$ & $1$ \\
	$33$ & $2$ & $6$ & $2$ & $4$ & $2$ \\
	$34$ & $3$ & $3$ & $1$ & $2$ & $5$ \\
	$35$ & $6$ & $4$ & $5$ & $6$ & $4$ \\
	$36$ & $5$ & $5$ & $4$ & $1$ & $3$ \\
	\bottomrule
	\end{tabular}
\end{minipage}
\medskip

As shown in the example, the Skating System allows for quite unintuitive results, making a formalization and an extensive axiomatic analysis particularly important to either justify or dismiss its usage both in ballroom tournaments and as a general voting rule.
Until today, there are only a handful of investigations of the Skating System~\cite{mor:t:the-skating-system,mor:t:improving-the-skating-system-part-two}. Therefore, the Skating System is mostly unknown in the area of COMSOC.

We introduce a new voting rule---called \emph{Skating System Single} (\emph{SkS}, for short)---that corresponds to and formalizes parts of the rule set of the Skating System, namely the winner determination of each dance in the final round.
To the best of our knowledge, the Skating System has not been formalized as a social choice function yet.
However, such a formalization is necessary to be able to perform axiomatic or complexity-theoretic analyses.
Since the rule set of the Skating System is already widely used in real-life tournaments worldwide,\footnote{Specifically, the Skating System is used by the World DanceSport Federation and the International Olympic Committee.}
we deem it
necessary to start such studies.
Therefore, we provide an extensive axiomatic analysis of SkS,
considering properties such as the \emph{Condorcet criterion}, \emph{majority}, \emph{positive responsiveness}, \emph{(strong) monotonicity}, \emph{independence of clones},
and \emph{participation}.

Moreover, we study the complexity of \emph{manipulation} for SkS:
In an election, instead of all voters casting their votes sincerely when ranking the candidates from most to least preferred, some voters might also vote \emph{strategically} to change the election outcome in their favor.
That is, based on complete knowledge of all sincere voters' preferences, they might cast ballots that misrepresent their honest preferences with the intention of making a favored candidate the winner of an election (\emph{constructive case}) or prevent a despised candidate from becoming the winner (\emph{destructive case}).
A voting system is called \emph{strategy-proof} if no voter can improve the election outcome by voting strategically.
However, Gibbard~\cite{gib:j:manipulation-voting-schemes} and Satterthwaite~\cite{sat:j:strategy-proofness} showed that every voting system that satisfies certain reasonable axioms like \emph{nondictatorship} cannot be strategy-proof.
Bartholdi et al.~\cite{bar-orl:j:polsci:strategic-voting,bar-tov-tri:j:manipulating} introduced and studied constructive manipulation by a single manipulator who aims at making a favored candidate win an unweighted election.
Conitzer et al.~\cite{con-san-lan:j:when-hard-to-manipulate} added the destructive goal, generalized this for coalitions of manipulators and also considered weighted elections, i.e., elections where the contribution of each vote is weighted by some integer.
In their book chapter, Conitzer and Walsh~\cite{con-wal:b:handbook-comsoc-manipulation} survey the computational barriers to manipulation in voting.

Another type of attack on elections, for which we initiate the study of SkS in this work, is \emph{electoral control}, introduced by Bartholdi et al.~\cite{bar-tov-tri:j:control} in the constructive and Hemaspaandra et al.~\cite{hem-hem-rot:j:destructive-control} in the destructive case.
In control scenarios, an external agent, also called the \emph{(election) chair}, attempts to make a distinguished candidate win (constructive case) or not win (destructive case) an election by interfering with the structure of the election, e.g., by adding, deleting, or partitioning voters and candidates.

If the corresponding decision problem is \np-hard, we say that the voting rule is \emph{resistant} to that kind of attack. 
Otherwise, if the decision problem can be solved in polynomial time, we call this voting rule \emph{vulnerable} to that attack type.

\paragraph{Related Work.} Manipulation and control are also closely related to 
bribery~\cite{fal-hem-hem:j:bribery}, a third attack type on elections, surveyed together with control by Faliszewski and Rothe~\cite{fal-rot:b:handbook-comsoc-control-and-bribery}.
In bribery scenarios, an external agent tries to make a given candidate win an election (in the constructive case) or, not win (in the destructive case) by bribing voters to change their vote within a given budget.
The book chapter by Baumeister and Rothe~\cite{bau-rot:b-2nd-edition:economics-and-computation-preference-aggregation-by-voting} covers all three types of attacks on elections.

The newly introduced SkS is most closely related to the \emph{Bucklin} voting system.
For Bucklin voting, manipulation and bribery have been studied by Faliszewski et al.~\cite{fal-rei-rot-sch:j:manipulation-bribery-campaign-management-in-bucklin-fallback-voting} and electoral control by Erd\'{e}lyi et al.~\cite{erd-fel-rot-sch:j:control-in-bucklin-and-fallback-voting-theoretical,erd-fel-rot-sch:j:control-in-bucklin-and-fallback-voting-experimental}.
%
%
While the work of Bartholdi et al.~\cite{bar-orl:j:polsci:strategic-voting,bar-tov-tri:j:manipulating} marks the very beginnings of the area of COMSOC, this---and, in particular, the study of strategic voting, manipulation, and control of elections---is still a very active field of current research (see, e.g., \cite{egg-now:j:susceptibility-strategic-voting-pluratlity-and-instant-runoff,elk-obr-teh:c:temporal-fairness-multiwinner-voting,mau-nic-nue-rot-see:c:toward-completing-the-picture-of-control-in-schulze-and-ranked-pairs-elections}).
%

\section{Preliminaries}
\label{sec:preliminaries}

Before we turn to elections, 
we introduce some general assumptions and notation. 
For a logical expression~$P$, we define $\eins_P$ to be $1$ if $P$ is true, and $0$ otherwise. 
For $n \in \mathbb{N}\setminus \{0\}$, we write $[n]$ for the
set $\{1, \ldots, n\}$.
We assume that the reader is familiar with the basic concepts of complexity theory (for more background, the textbooks by see Papadimitriou~\cite{pap:b:complexity} and Rothe~\cite{rot:b:cryptocomplexity}).

An election is a pair $(C,V)$, where $C=\{c_1, \dots, c_m\}$ for $m\in \N$ is a set of candidates
and $V=(v_1, \dots, v_n)$ for $n\in \N$ is a list of votes specifying the voters' preferences over the candidates in~$C$.
Voters can express their preferences over the candidates in many different ways.
As is most common in social choice theory, SkS requires rankings of the candidates, i.e., each vote is given by a (strict) linear order over~$C$.
We write $c \succ_v d$ if a voter $v$ prefers candidate $c$ to candidate~$d$, omitting the subscript when $v$ is clear from the context: $c \succ d$.
We use $\mathrm{pos}_v(c)$ to denote the position of candidate $c$ in vote~$v$.
A voting system maps any given election $(C,V)$ to a subset of~$C$, the \emph{winner(s) of $(C,V)$}.

For $i \in [m]$ and $c\in C$, let
\[
\scorei{i}(c) = \sum_{v\in V}\eins_{\mathrm{pos}_v(c) \leq i}
\]
be the number of votes in which $c$ is ranked among the top-$i$ positions in the votes from~$V$.
Further, let
\[
\sumposi{i}(c) = \sum_{v\in V}\mathrm{pos}_v(c) \cdot \eins_{\mathrm{pos}_v(c) \leq i}
\]
be the sum of positions candidate $c$ achieves up to stage~$i$ in the votes~$V$.
Only if needed, we use $(C,V)$ as a subscript and explicitly write $\scorei{i}_{(C,V)}(c)$ and $\sumposi{i}_{(C,V)}(c)$.

We now define the well-known \emph{Bucklin voting system} which, given an election $(C,V)$, proceeds in stages until at least one candidate has reached (or exceeded) the \emph{majority threshold} $\maj(V) = \floor*{\nicefrac{\abs{V}}{2}} + 1$, i.e., is among at least $\maj(V)$ top-$i$ positions in the votes~$V$.
That is, in each stage~$i$, we check if at least one candidate reaches the majority threshold.
If none of the candidates reached it in the current stage, we proceed to the next stage.
For any candidate $c \in C$ and for any stage $i \in [m]$, with $m = |C|$, the \emph{Bucklin score of $c$ in stage $i$} is defined by $\BucklinScorei{i}(c) = \scorei{i}(c)$.
Let $i^*$ 
be the smallest stage in which at least one candidate reaches $\maj(V)$,
and return the candidate(s) with the highest Bucklin score in this stage~$i^*$ as the (set of) \emph{Bucklin winner(s)}, 

Obviously, every Bucklin winner who is already found in the first stage 
must be a unique winner (i.e., wins alone), and if we were to proceed to the very last stage $m = \abs{C}$, all candidates in $C$ are Bucklin winners.
In particular, we always have at least one Bucklin winner.

\medskip
\noindent
\begin{minipage}[h!]{0.7\columnwidth}
	\begin{example}\label{ex:bucklin}
			Let $(C,V)$ be an election with $C = \{U,X,Y,Z\}$ and the following votes shown on the right.
			In stage $i=1$, we have $\BucklinScorei{1}(U) =2$ and $\BucklinScorei{1}(c) =1$ for $c \in \{X,Y,Z\}$.
			None of the candidates reaches the majority threshold $\maj(V) = 3$; thus we proceed with stage~2 in which both $U$ and $X$ reach this threshold with a score of~$3$.
			Hence, both $U$ and $X$ are the Bucklin winners of the election.
	\end{example}
\end{minipage}
\hfill
\begin{minipage}[t]{0.28\columnwidth}
	\begin{tabularx}{\linewidth}{r C}
		\toprule
		Anne: & $U \succ X \succ Y \succ Z$ \\
		Bob:  & $U \succ X \succ Y \succ Z$ \\
		Caro: & $X \succ U \succ Y \succ Z$ \\
		Dave: & $Y \succ Z \succ U \succ X$ \\
		Ella: & $Z \succ Y \succ U \succ X$ \\
		\bottomrule
	\end{tabularx}
\end{minipage}

\section{Skating System Single}
\label{sec:SkS}

Example~\ref{ex:bucklin} shows a Bucklin election with multiple winners.
A keen reader may have noticed that, even though both winners share the same Bucklin score, one of the candidates (namely $U$) is ranked higher than the other in two of the three decisive votes (and four of the five voters in total).
While in this example the difference in the sum of positions is only one, we later show in Theorem~\ref{thm:properties-sks-bucklin} that this difference can be arbitrarily large.
Surely, this makes a strong case that $U$ deserves to win alone, does it not?

In this section, we---first intuitively and then more formally---introduce the \emph{Skating System Single} (\emph{SkS}, for short), a new voting system inspired by the rule set of the Skating System~\cite{williams2018skating}.
Just as Bucklin voting, SkS proceeds in stages.
While proceeding through each stage, various conditions are checked until a winner is found.
Essentially, most of the conditions can be seen as a form of tie-breaking if more than one candidate qualifies as a possible winner.
However, even though this intrinsic tie-breaking exists, SkS is not a resolute rule.
Before formally defining winner determination, let us first give an overview of the stages in SkS, given an election $(C,V)$ with $m = \abs{C}$ candidates.

\begin{description}
\item[Stage 1:]
  Is any candidate $c \in C$ in the top position of at least $\maj(V)$ votes?
\begin{itemize}
	\item [\emph{Yes:}] Candidate $c$ is the unique SkS winner.
	\item [\emph{No:}] Continue with stage~2.
\end{itemize}

\item[Stage~$\pmb{i,\ 2 \leq i \leq m}$:]
  Are there any candidates $C' \subseteq C$ in the top-$i$ positions of at least $\maj(V)$ votes?
\begin{itemize}
\item [\emph{Yes:}] Let $C^* \subseteq C'$ be the candidates from $C'$ with the highest score (to be formally defined below).
  Is it true that $\abs{C^*}$ = 1?
	\begin{itemize}
		\item [\emph{Yes:}] This candidate in $C^*$ is the unique SkS winner.
		\item [\emph{No:}] Does a single candidate $c \in C^*$ have the lowest sum of top-$i$ positions (to be formally defined below)?
		\begin{itemize}
			\item [\emph{Yes:}] This candidate $c$ is the unique SkS winner.
			\item [\emph{No:}] If $i=m$ (i.e., currently we are already in the last stage), return all candidates $C^*$ as SkS winners.

        Otherwise, continue with stage~$i+1$, but only with candidates from $C^*$ who tie for the lowest sum of
        top-$i$ positions.\footnote{Proceeding with only the candidates from $C^*$ who tie for the lowest sum of top-$i$ positions instead of all candidates from $C^*$ leads to an interesting interaction.
        A candidate $c \in C^*$ may be excluded based on the sum of positions, even if $c$ would beat all candidates who tie for the lowest sum of positions by score in the next stage.
        }
		\end{itemize}
	\end{itemize}
	\item [\emph{No:}] Continue with stage~$i+1$ (with all candidates).
\end{itemize}
\end{description}

Having a first idea of what the complete SkS voting system looks like, we can now start with its formally precise definition.

Again, we use the majority threshold $\maj(V) = \floor*{\nicefrac{\abs{V}}{2}} + 1$.
Next, we define the scoring functions used by SkS.
In each
stage, the score of a candidate will consist of two components: First, the number of votes in which the candidate was ranked in the top-$i$ positions and, second, the sum of these positions, allowing for a more refined distinction of candidates.
The higher a candidate is ranked in a voter's preference, the less impact this vote has on the sum of positions, i.e., a lower ranking in the top-$i$ positions will increase the sum of positions more than being the top candidate in a vote, so a lower sum of positions is favorable.

\begin{definition}[\textbf{SkS}]\label{def:SkS} 
	Let $(C,V)$ be an election with $m = \abs{C}$ candidates.
	%
	For any candidate~$c$, let
        \[
        i^*_c = \min_{i \in [m]}\{i \mid \scorei{i}(c) \geq \maj(V)\}
        \]
        be the minimal stage in which $c$ reaches the majority threshold, and let $i^* = \min_{c \in C}(i^*_c)$ be the smallest stage in which at least one candidate reaches the majority threshold.
	The SkS voting rule looks for the minimal stage $j \in[i^*, m]$ in which one single candidate has the highest score.
	To this end, consider the following cases for the set
		$C_j = \argmax_{c \in C} (\scorei{j}(c))$.
	\begin{description}
	\item[Case $\pmb{\abs{C_j} = 1}$:]
		  The sole candidate in $C_j$ is the unique SkS winner.
		\smallskip

	\item[Case $\pmb{\abs{C_j} > 1}$:]
		  SkS compares the sum of positions for all candidates in $C_j$:
		  $C_j' = \argmin_{c' \in C_j} (\sumposi{j}(c'))$.
		\begin{description}
		\item[Case $\pmb{\abs{C_j'} = 1}$:]
				  The sole candidate in $C_j'$ is the unique SkS winner.
		\smallskip

		\item[Case $\pmb{\abs{C_j'} > 1}$:]
				  If more than one candidate has the lowest sum of positions, we proceed to the next stage $j+1$ and repeat the steps mentioned above with only the candidates in $C_j'$ (i.e., set $C = C_j'$).
				  As soon as $j = m$, all candidates $C_j'$ are the SkS winners.
		\end{description}
	\end{description}
\end{definition}

The following example illustrates how SkS works.

\begin{example}
\label{ex:sks}
	Let $(C,V)$ be an election with the set of candidates $C = \{X,Y,Z\}$ and the voters Anne, Bob,
		Caro, and Dave having the following preferences: Anne: $X \succ Y \succ Z$, Bob: $Y \succ X \succ Z$, Caro: $Y \succ Z \succ X$, and Dave: $Z \succ X \succ Y$.
		We have $\maj(V)= \floor*{\nicefrac{4}{2}} + 1 = 3$.
		The stages of the SkS election are depicted in Table~\ref{tab:example-sks}.
		Candidates within the top-$i$ positions of the votes in stage~$i$ are marked by a~$\checkmark$.
		The two rows at the bottom of the tables show the resulting scores.

	\begin{table}[t!]
		\begin{subtable}[c]{0.48\textwidth}
	\begin{tabularx}{\linewidth}{rl|CCC}
		\toprule
		\textbf{Stage 1} & Votes & $X$ & $Y$ & $Z$\\
		\midrule
		Anne: &$\color{red} X \color{black} \succ Y \succ Z$ & $\checkmark$ & - & -\\
		Bob:  &$\color{red} Y \color{black} \succ X \succ Z$ & - & $\checkmark$ & -\\
		Caro: &$\color{red} Y \color{black} \succ Z \succ X$ & - & $\checkmark$ & -\\
		Dave: &$\color{red} Z \color{black} \succ X \succ Y$ & - & - & $\checkmark$\\
		\midrule
		& \multicolumn{1}{r|}{$\scorei{1}(c)$} 				 & 1 & 2 & 1 \\
		& \multicolumn{1}{r|}{$\sumposi{1}(c)$} 			 & 1 & 2 & 1\\
		\bottomrule
	\end{tabularx}
	\subcaption{Stage~1 of the election.}\label{tab-a:example-sks-stage-one}
	\end{subtable}
	\hfill
	\begin{subtable}[c]{0.48\textwidth}
	\begin{tabularx}{\linewidth}{rl|CCC}
		\toprule
		\textbf{Stage 2} & Votes & $X$ & $Y$ & $Z$\\
		\midrule
		Anne: &$\color{red} X \succ Y \color{black} \succ Z$ & $\checkmark(1)$ & $\checkmark(2)$ & -\\
		Bob:  &$\color{red} Y \succ X \color{black} \succ Z$ & $\checkmark(2)$ & $\checkmark(1)$ & -\\
		Caro: &$\color{red} Y \succ Z \color{black} \succ X$ & - & $\checkmark(1)$ & $\checkmark(2)$\\
		Dave: &$\color{red} Z \succ X \color{black} \succ Y$ & $\checkmark(2)$ & - & $\checkmark(1)$\\
		\midrule
		& \multicolumn{1}{r|}{$\scorei{2}(c)$} 				 & 3 & 3 & 2 \\
		& \multicolumn{1}{r|}{$\sumposi{2}(c)$} 			 & 5 & 4 & 3\\
		\bottomrule
	\end{tabularx}
	\subcaption{Stage~2 of the election.}\label{tab-b:example-sks-stage-two}
	\end{subtable}
	\caption{Stages of the SkS election in Example~\ref{ex:sks}.}\label{tab:example-sks}
\end{table}

	Since in the first stage, no candidate receives enough votes to reach the majority threshold, the election proceeds to the second stage. Here, 
	for each candidate, we also show in which position of the given vote the candidate is ranked in case this candidate was ranked within the top-$i$ positions. That is, we write $\checkmark(2)$, if $X$ is ranked in the second position.
	In the second stage, both $X$ and $Y$ reach the majority threshold.
	Since candidate $Y$ has a lower sum of positions, $Y$ is the unique SkS winner.
\end{example}


As mentioned above, SkS resembles Bucklin voting: They both proceed in stages and use the majority principle.
Therefore, let us compare them in some more detail and discuss their similarities and differences.
The first and most obvious similarity is that the first part of the SkS score equals the Bucklin score:
$\scorei{i}(c) = \BucklinScorei{i}(c) =  \sum_{v\in V}\eins_{\mathrm{pos}_v(c) \leq i}$.
Second, the first stage of both voting systems coincides: A candidate $c$ with $\scorei{1}(c) \ge \maj(V)$ wins alone in both SkS and Bucklin voting.
Beyond the obvious, unique Bucklin winners reaching the majority threshold at a later stage will also be unique SkS winners.
The reverse, however, is not always the case:
Even if an SkS winner is unique, there can be multiple Bucklin winners of the election.
For example, in the election from Example~\ref{ex:sks}, we have seen that $Y$ is the unique SkS winner, but it is easy to see that both $X$ and $Y$ would be the Bucklin winners.
Furthermore, all SkS winners are also Bucklin winners; again, the reverse is not always the case
(as we have just seen).
We summarize and prove these and other properties below.


\begin{theorem}
  \label{thm:properties-sks-bucklin}
  \begin{enumerate}[label=(\roman*)]
    \item Every unique Bucklin winner is a unique SkS winner.
    \item Every SkS winner is a Bucklin winner.
    \item For all $k \ge 1$, there is an election with the same $k$ SkS and Bucklin winners.
    \item The difference in the sum of positions of Bucklin winners can be arbitrarily large in the decisive stage.
  \end{enumerate}
\end{theorem}

\begin{proof}

 \begin{enumerate}[label=(\roman*)]
  \item Assume that $c_w \in C$ is the unique Bucklin winner of an election $(C,V)$.
  Let $i^*$ be the minimal stage in which at least one candidate reaches the majority threshold in Bucklin voting.
  Clearly, $c_w$ reaches the majority threshold in stage~$i^*$.
  The Bucklin winner is defined as $\argmax_{c^* \in C} (\BucklinScorei{i^*}(c^*))$.
  Since $c_w$ is the sole Bucklin winner, $c_w$ has a greater Bucklin score than all other candidates in stage~$i^*$.
  The majority threshold is the same for both SkS and Bucklin voting and, thus, we know that $c_w$ also reaches the majority threshold in an SkS election in stage~$i^*$ and stage $i^*$ is, again, the smallest stage in which a candidate has reached the majority threshold.
  By definition of $\scorei{i}(c) =  \sum_{v\in V}\eins_{\mathrm{pos}_v(c) \leq i} = \BucklinScorei{i}(c)$, it immediately follows that $\argmax_{c \in C} (\scorei{i^*}(c)) = \{c_w\}$, and therefore, $c_w$ is also the unique SkS winner of $(C,V)$.
  \smallskip

  \item Let $(C,V)$ be an election.
  First, if $(C,V)$ has a unique Bucklin winner, we know from
  the first statement that the SkS and Bucklin winners of $(C,V)$ are the same.
  So assume that $(C,V)$ has two or more Bucklin winners.
  It remains to show that every SkS winner is also a Bucklin winner in this case.
  We show this by contradiction.
  Let $C_w^B$ be the set of Bucklin winners and $C_w^S$ the set of SkS winners of $(C,V)$.
  For the sake of contradiction, assume that $C_w^S \not\subseteq C_w^B$, i.e., at least one SkS winner, say $c^*$, is not a Bucklin winner.
  Since $c^*$ is an SkS winner, we have $\scorei{i^*}(c^*) \ge \maj(V)$, where $i^*$ is the minimal stage in which any candidate reached the majority threshold.
  Since $c^*$ is not a Bucklin winner, we also know that $\BucklinScorei{i^*}(c^*) < \BucklinScorei{i^*}(c)$ and $\BucklinScorei{i^*}(c) \ge \maj(V)$ for all $c$ in $C_w^B$.
  Immediately, we have $\scorei{i^*}(c^*) < \scorei{i^*}(c)$ for all $c$ in $C_w^B$.
  Since $c^* \in C_w^S$, we have $c^* \in \argmax_{c \in C} (\scorei{i^*}(c))$.
  This is a contradiction.
  \smallskip

  \item For $k = 1$, this directly follows from the first statement.
  For $k > 1$, construct an election $(C,V)$ as follows.
  Let $C = \{c_1, c_2, \ldots, c_k\}$.
  Construct $k$ votes by setting
  \[
   \begin{array}{ccccccccc}
    c_1 & \succ & c_2 & \succ & \cdots & \succ & c_{k-1} & \succ & c_k\\
    c_2 & \succ & c_3 & \succ & \cdots & \succ & c_k    & \succ & c_1\\
    &       &     &      &  \vdots &       &       &       &    \\
    c_k & \succ & c_1 & \succ & \cdots & \succ & c_{k-2} & \succ & c_{k-1}
   \end{array}
  \]

  Due to the cyclic permutation of the votes, each candidate is ranked at each position exactly once and receives exactly one additional point in each stage.
  Therefore, all candidates reach the majority threshold $\maj(V)$ at the same time: in stage
  $\frac{k+1}{2}$ for odd $k$, and in stage
  $\frac{k}{2} + 1$ for even~$k$.
  Since all candidates tie in their scores, they all are Bucklin winners.
  For SkS, the positional scores of all candidates in each stage $i$ is $\sum_{\ell \in [i]}{\ell}$ and, as in Bucklin voting, all candidates tie in each stage.
  Therefore, SkS proceeds through all stages without excluding any candidate and, finally, declares all of them winners.
  \smallskip

  \item We give a construction for an election in which the difference in the sum of positions between two Bucklin winners can be controlled by setting two integer parameters.
  Let $i, n \in \mathbb{N}$ be two integers with $i \ge 3$ and $n \ge 0$ and let $C' = \{c_1, \ldots, c_{2n}\}$ and $D = \{d_1, \ldots, d_{n+2}\}$ be two sets of candidates.
  Construct an election $(C,V)$ with the candidates $C = \{a,b\} \cup C' \cup D$ and the following votes---note that we only depict the first $n+2$ positions of each vote, as the remainder of the vote can be arbitrary:
  \smallskip
  \begin{center}
   \begin{tabularx}{0.5\textwidth}{ll}
    \toprule
    \multicolumn{1}{c}{Vote list $V$} & \\
    \midrule
    $a\phantom{_{1}} \succ c_1\phantom{_{+1}} \succ \ldots \succ c_n\phantom{_{+1}}  \succ b\phantom{_{n+2}}  \succ \ldots $ & \rdelim\}{3}{3mm}[$(i-1)$] \tabularnewline
    \multicolumn{1}{c}{\vdots} & \tabularnewline
    $a\phantom{_{1}} \succ c_{1}\phantom{_{+1}} \succ \ldots \succ c_n\phantom{_{+1}}  \succ b\phantom{_{n+2}}  \succ \ldots $ & \tabularnewline
    $b\phantom{_{1}} \succ c_{n+1} \succ \ldots \succ c_{2n}\phantom{_{+}}  \succ a\phantom{_{n+2}}  \succ \ldots $ & \tabularnewline
    $d_1 \succ d_2\phantom{_{+1}} \succ \ldots \succ d_{n+1} \succ d_{n+2} \succ \ldots $ & \rdelim\}{3}{3mm}[$(i-1)$]  \tabularnewline
    \multicolumn{1}{c}{\vdots} & \tabularnewline
    $d_1 \succ d_2\phantom{_{+1}} \succ \ldots \succ d_{n+1} \succ d_{n+2} \succ \ldots $ & \tabularnewline
    \bottomrule
   \end{tabularx}
  \end{center}
  \smallskip
  We have $\maj(V) = \floor*{i-\frac{1}{2}} + 1 = i$.
  It is clear, that each candidate from $C'$ and $D$ occurs at most $i-1$ times in the top $n+2$ positions.
  Candidate $a$ is ranked in the first postion a total of $i-1$ times and then once again in position $n+2$,
  reaching the majority threshold in stage $n+2$.
  The inverse is true for candidate $b$, who is ranked only once in the first position and then $i-1$ times in position $n+2$, also
   reaching the majority threshold in stage $n+2$.
  Thus both $a$ and $b$ are the Bucklin winners.
  The difference between the two winners is immediately clear, though:
  Calculating the sum of positions in stage $n+2$ for $a$ and~$b$, we have
  \begin{align*}
   \sumposi{n+2}(a) &= i-1 + n+2, \\
   \sumposi{n+2}(b) &= 1 + (i-1)(n+2).
  \end{align*}
  By increasing either~$i$ (and thus the number of votes) or $n$ (and thus the number of candidates) or both, we can
  boost the difference in the sum of positions between $a$ and $b$ arbitrarily.\qedhere
 \end{enumerate}
\end{proof}



SkS and Bucklin voting differ in their behavior whenever some candidates reach the majority threshold due to the same number of votes.
In such cases, Bucklin returns all these candidates as winners, yet SkS instead first compares their sums of positions: If they are the same, it moves on to the next stage with the relevant candidates, repeating the procedure either until one candidate is the sole winner or until the last stage reached with multiple SkS winners.
Following this, we state the following lemma on winner determination in SkS.
Intuitively, Lemma~\ref{lem:winner-sks-ties-restricted-no-sum-pos} states that once candidates have tied both in terms of the highest score and lowest sum of positions (while also reaching the majority threshold) leading to a reduced candidate set in the next stage, any unique winner beats all other remaining candidates by a strictly higher score.

\begin{lemma}\label{lem:winner-sks-ties-restricted-no-sum-pos}
Let $(C,V)$ be an election and $C' \subseteq C$ some candidates.
Assume that, for some stage~$i$, we have $\scorei{i}(c) = \scorei{i}(d) \geq \maj(V)$, $\scorei{i}(c) > \scorei{i}(x)$, and $\sumposi{i}(c) = \sumposi{i}(d)$ for all $c,d \in C',c \neq d$, and $x \in (C\setminus C')$, i.e., SkS proceeds to stage $i+1$ with the reduced candidate set~$C'$.
Then it holds that
\begin{enumerate}[label=(\roman*)]
 \item for a unique SkS winner~$c$, we have $\scorei{j}(c) > \scorei{j}(d)$ for some $i < j \leq \abs{C}$ and $d \in (C'\setminus \{c\})$;
 \item if there is no unique SkS winner, we have $\scorei{j}(c) = \scorei{j}(d)$ for all $i < j \leq \abs{C}$ and for some $c \in C', d \in C' \setminus \{ c \}$.
\end{enumerate}
\end{lemma}

\begin{proof}
 \begin{enumerate}[label=(\roman*)]
  \item Let SkS be in stage $i$ and assume we have a reduced candidate set $C' \subseteq C$.
  Since we have a reduced candidate set, all candidates from the reduced set must have tied both the score and sum of positions in the previous round.
  Assume that $c$ is the first unique SkS winner in stage $j$ with $i < j \leq \abs{C}$ but $\scorei{j}(c) \leq \scorei{j}(d)$ for some $d \in C^*  \subseteq (C'\setminus \{c\})$.
  Since $c$ is a winner there can be no candidate $d$ with a higher score and thus we have $\scorei{j}(c) = \scorei{j}(d)$.
  Additionally, since $c$ is the unique SkS winner, the sum of positions for $c$ is the uniquely lowest sum of positions among all candidates from $C^*$.
  Since all candidates from $C^*$ must have tied both in highest score and sum of positions in stage $j-1$, we have that $c$ either decreased their sum of positions or gained the same increase in score as the other candidates while simultaneously receiving less sum of positions. As both are not possible, we thus have a contradiction.
  
  \item This follows directly from the definition of SkS.\qedhere
 \end{enumerate}
\end{proof}


\section{Axiomatic Analysis of SkS}
\label{sec:axiomatic-analysis-of-SkS}

In this section, we consider the following axiomatic properties of voting systems:
the Condorcet criterion,
the majority criterion,
monotonicity,
positive responsiveness,
strong monotonicity,
independence of irrelevant alternatives,
independence of clones,
consistency,
participation,
nondictatorship,
citizens' sovereignty,
resoluteness, and
strategy-proofness.
These have been motivated and discussed in depth in the social choice literature  and have been investigated for an abundance of voting systems (see, e.g., \cite{arr-sen-suz:b:handbook-social-choice-and-welfare-volume-2,bau-rot:b-2nd-edition:economics-and-computation-preference-aggregation-by-voting,tid:j:independence-of-clones,zwi:b:handbook-comsoc-introduction-to-voting-theory}).
In the following, we study which of them are satisfied by SkS and which are not.

\begin{theorem}
  \label{thm:axioms-for-SkS}
    SkS satisfies
    the majority criterion,
    positive responsiveness,
    monotonicity,
    nondictatorship, and
    citizens' sovereignty.
    On the other hand, 
    SkS does not satisfy
    the Condorcet criterion,
    strong monotonicity,
    independence of irrelevant alternatives,
    independence of clones,
    consistency,
    participation,
    resoluteness, nor
    strategy-proofness.
\end{theorem}

\begin{proof}
 Due to space limitations, among the axioms satisfied by SkS, we will only present the proofs for monotonicity and positive responsiveness,
 two closely related, well-known axioms from the social choice literature.
 Recall that a voting system is \emph{monotonic} if for every election $(C,V)$ and for each candidate $c \in C$,
 if
 \begin{enumerate*}[label=(\roman*)]
    \item $c$ wins the election, and
    \item a new election $(C,V')$ is generated from $(C,V)$ by improving the position of $c$ in some votes, leaving the relative rankings of all other candidates unchanged, then $c$ also wins this new election.
 \end{enumerate*}
 Further, \emph{positive responsiveness} means that a (possibly tied) winner $c$ turns a unique winner whenever $c$ is raised in some votes, again leaving everything else the same.

\smallskip\noindent\textbf{Monotonicity:}
  Assume a candidate $c$ wins an SkS election $(C,V)$.
  Now, we improve the positions of $c$ in some votes, leaving the relative rankings of all other candidates unchanged.
  Let $W$ be the resulting list of new votes.
  For the sake of contradiction, assume that $c$ does no longer win the SkS election $(C,W)$.
  However, since $c$ has only improved the positions in some votes, if $c$ has won $(C,V)$ already in the first stage due to a majority of top positions, $c$ must also win $(C,W)$ in the first stage due to a majority of top positions.
  If $c$ has won $(C,V)$ at a later stage, let $i^{*}_c > 1$ be the stage where $c$ meets or exceeds the majority threshold in $(C,V)$.
  Now, since $c$ has only improved the positions in some votes while none of the other candidates can have a better position in the votes of $W$ than in those of~$V$, it follows that
  \begin{enumerate*}[label=(\roman*)]
    \item $c$ reaches the majority threshold at a stage~$j^{*}_c \leq i^{*}_c$ in $(C,W)$ and no other candidate can reach the majority threshold at an earlier stage than in $(C,V)$,
    \item $\scorei{j^{*}_c}_{(C,W)}(c) \geq \scorei{i^{*}_c}_{(C,V)}(c)$,
      and
    \item $\sumposi{j^{*}_c}_{(C,W)}(c) \leq \sumposi{i^{*}_c}_{(C,V)}(c)$.
   \end{enumerate*}
  Hence, $c$ again is an SkS winner in $(C,W)$, a contradiction.
  
\smallskip\noindent\textbf{Positive responsiveness:}
  Let $(C,V)$ be an election and let $c,d \in C$ be two SkS winners of $(C,V)$.
  Recall from Lemma~\ref{lem:winner-sks-ties-restricted-no-sum-pos} that if multiple winners exist in an SkS election, then once the majority threshold is reached at some stage $i$, candidates $c$ and $d$ will continue to tie in both score and sum of positions for every following stage.
  Let $V'$ be a vote list derived from $V$ by improving, without loss of generality, $c$'s position in at least one ballot, while leaving the relative rankings of all other candidates unchanged.
  Assume that $c$ is
  not a unique winner of election $(C,V')$.
  Then either
  \begin{enumerate*}[label=(\roman*)]%
   \item $c$ does not win at all or
   \item $c$ does win alongside someone else, say $x \in C$.
  \end{enumerate*}
  Since SkS satisfies monotonicity, (i)~immediately leads to a contradiction, which leaves us with case~(ii).
  Due to $c$'s improvement, there must be some stage $j$ with $1 \leq j \leq \abs{C}$ where $\scorei{j}_{V'}(c) > \scorei{j}_{V}(c)$ and $\sumposi{i}_{V'}(c) < \sumposi{i}_{V}(c)$.
  Note that for any other candidate $y \in C\setminus \{c\}$ and every stage $i$, we have $\scorei{i}_{V'}(y) \leq \scorei{i}_{V}(y)$ and $\sumposi{i}_{V'}(y) \geq \sumposi{i}_{V}(y)$.
  In particular, $x$'s position did not improve.
  Let $i' \leq i$ be the stage in which $c$ reaches the majority threshold in $(C,V')$.
  If $i' < i$, we have that $c$ reached the majority threshold earlier than in the original election.
  Since $x$ is also a winner, we have $\scorei{i'}_{V'}(c) = \scorei{i'}_{V'}(x)$,
  a contradiction.
  Lastly, if $i' = i$, due to Lemma~\ref{lem:winner-sks-ties-restricted-no-sum-pos}, we have $\scorei{j}_{V'}(c) = \scorei{j}_{V'}(x)$ and $\sumposi{j}_{V'}(c) = \sumposi{j}_{V'}(y)$ for $i' \leq j \leq \abs{C}$, again a contradiction.
%

%
	\begin{table}[t!]
		\begin{subtable}[c]{0.5\columnwidth}
			\centering
			\begin{tabularx}{0.6\linewidth}{rC}
				\toprule
				Anne: & $X \succ Y \succ Z$ \\
				Bob:  & $X \succ Y \succ Z$ \\
				Caro: & $Y \succ X \succ Z$ \\
				Dave: & $Y \succ Z \succ X$ \\
				Ella: & $Z \succ Y \succ X$ \\
				Fred: & $Z \succ Y \succ X$ \\
				Gail: & $Z \succ Y \succ X$ \\
				\bottomrule
			\end{tabularx}
			\subcaption{Vote list $V$ of the original election.}\label{tab-a:votes-original-election-in-proof-thm-axioms-not-satisfied}
		\end{subtable}
		\begin{subtable}[c]{0.5\columnwidth}
			\centering
			\begin{tabularx}{0.6\linewidth}{rC}
				\toprule
				Anne: & $X \succ Y \succ Z \succ Z'$ \\
				Bob:  & $X \succ Y \succ Z \succ Z'$ \\
				Caro: & $Y \succ X \succ Z \succ Z'$ \\
				Dave: & $Y \succ Z \succ Z' \succ X$ \\
				Ella: & $Z \succ Z' \succ Y \succ X$ \\
				Fred: & $Z \succ Z' \succ Y \succ X$ \\
				Gail: & $Z \succ Z' \succ Y \succ X$ \\
				\bottomrule
			\end{tabularx}
		\subcaption{Vote list $V'$ after adding the clone $Z'$.}\label{tab-b:votes-after-clone-added-in-proof-thm-axioms-not-satisfied}
		\end{subtable}
		\caption{Vote lists of the two elections in the proof of Theorem~\ref{thm:axioms-for-SkS}.}\label{tab:votes-in-proof-thm-axioms-not-satisfied}
	\end{table}
        
\smallskip\noindent\textbf{Condorcet criterion, independence of clones, and independence of irrelevant alternatives:}
        Again, due to space limitations, we will present just one counterexample in Table~\ref{tab:votes-in-proof-thm-axioms-not-satisfied} showing that SkS satisfies neither the Condorcet criterion nor independence of clones nor independence of irrelevant alternatives
        (referring to,
        e.g., \cite{arr-sen-suz:b:handbook-social-choice-and-welfare-volume-2,bau-rot:b-2nd-edition:economics-and-computation-preference-aggregation-by-voting,tid:j:independence-of-clones,zwi:b:handbook-comsoc-introduction-to-voting-theory} for the formal definitions of these properties).
	Let $C=\{X,Y,Z\}$ and $C'=\{X,Y,Z,Z'\}$.
	The list $V$ of votes over~$C$ is depicted in Table~\ref{tab-a:votes-original-election-in-proof-thm-axioms-not-satisfied}.
	Since $\abs{V} = 7$, we have $\maj(V) = 4$.
	Further, we have $\scorei{1}_{(C,V)}(c)=2$ for $c \in \{X,Y\}$ and $\scorei{1}_{(C,V)}(Z)=3$.
	Since neither candidate reaches the required threshold in stage~$i=1$, we continue with the next stage.
	Here, we have $\scorei{2}_{(C,V)}(X)=3$, $\scorei{2}_{(C,V)}(Z)=4$ and $\scorei{2}_{(C,V)}(Y)=7$.
	Now, $Y$ is the unique SkS winner because $\scorei{2}_{(C,V)}(c) \ge \maj(V)$ for $c\in \{Y,Z\}$, and
	$\scorei{2}_{(C,V)}(Y) > \scorei{2}_{(C,V)}(Z)$.
	Adding a clone $Z'$ of candidate $Z$ yields the election $(C',V')$ depicted in Table~\ref{tab-b:votes-after-clone-added-in-proof-thm-axioms-not-satisfied}.
	Note that the clone $Z'$ is always ranked in the position behind $Z$ and the rest of the votes is left untouched.
	The scores in the first stage of the election remain the same.
	However, in the second stage, the scores change with respect to the election $(C,V)$.
	That is, the scores of $X$ and $Z$ are unchanged, while $Y$ loses three points to the clone~$Z'$.
	Now, we have $\scorei{2}_{(C',V')}(c)=4$ for $c \in \{Y,Z\}$ and $\scorei{2}_{(C',V')}(d)=3$ for $d \in \{X,Z'\}$.
	Since
		$\maj(V') = \scorei{2}_{(C',V')}(Y) = \scorei{2}_{(C',V')}(Z)$,
	we have to consider the sum of positions for $Y$ and $Z$ to determine the winner.
	We have
	$\sumposi{2}_{(C',V')}(Z) = 5 < 6 = \sumposi{2}_{(C',V')}(Y)$.
	Thus, after the addition of the clone~$Z'$, the winner of the SkS election changed from $Y$ to~$Z$, and SkS satisfies neither independence of clones nor independence of irrelevant alternatives.
	Note that, in the second election, $Y$ is
        still a Condorcet winner and yet loses the
        SkS election to $Z$ now.
	Therefore, SkS does not satisfy the Condorcet criterion.
\end{proof}

As Bucklin and SkS are heavily related, their axiomatic properties often align.
Nonetheless, they are not always the same, as shown by positive responsiveness, which SkS satisfies while Bucklin does not.\footnote{%
 It is easy to see that Bucklin violates positive responsiveness by a small counterexample.
 Consider the candidates $C=\{a,b,c\}$ and the votes $a\succ b\succ c\succ d$, $b\succ a\succ c\succ d$, and $c\succ d\succ a\succ b$.
 Both $a$ and $b$ are the Bucklin winners in stage two.
 We can now improve $b$'s position in the last vote to $c\succ d\succ b\succ a$ without changing the election result.
}
A comprehensive overview of all properties considered here is given in Table~\ref{tab:properties}.

\begin{table}[h!]
 \centering
 \small
 \begin{tabularx}{\linewidth}{rLLLLLLLLLLLLLL}
  \toprule
  &
  \rotatebox{55}{Condorcet Criterion} &
  \rotatebox{55}{Majority} &
	\rotatebox{55}{Monotonicity} &
	\rotatebox{55}{Positive Responsiveness} &
  \rotatebox{55}{Strong Monotonicity} &
  \rotatebox{55}{\parbox{\widthof{Independence of Clones}}{Independence of Irrelevant Alternatives}} &
  \rotatebox{55}{Independence of Clones} &
  \rotatebox{55}{Consistency} &
  \rotatebox{55}{Participation} &
  \rotatebox{55}{Nondictatorship} &
  \rotatebox{55}{Citizens' Sovereignty} &
  \rotatebox{55}{Resoluteness} &
  \rotatebox{55}{Strategy-Proofness}&\\
  \midrule
  Bucklin &\xmark & \cmark  & \cmark & \xmark &  \xmark  & \multicolumn{1}{c}{\xmark} &\xmark & \xmark & \xmark & \cmark & \cmark & \xmark &\xmark &\\
  \midrule
  SkS &\xmark & \cmark  & \cmark & \cmark &  \xmark  & \multicolumn{1}{c}{\xmark} &\xmark & \xmark & \xmark & \cmark & \cmark & \xmark &\xmark &\\
  \bottomrule
 \end{tabularx}
 \normalfont
 \caption{Comparison of axiomatic properties between SkS and Bucklin (see, e.g., \cite{bau-rot:b-2nd-edition:economics-and-computation-preference-aggregation-by-voting}).}
 \label{tab:properties}
\end{table}

\section{Complexity of Manipulation in SkS}
\label{sec:complexity}

In this section, we will study the complexity of manipulation in SkS.
As mentioned in the introduction, Bartholdi et al.~\cite{bar-orl:j:polsci:strategic-voting,bar-tov-tri:j:manipulating} initiated the study of manipulating voting systems in terms of determining the computational complexity of the corresponding decision problems.
They considered the most restricted problem variant: constructive manipulation by a single manipulator seeking to make a preferred candidate win an unweighted election.

The most general problem, due to Conitzer et al.~\cite{con-san-lan:j:when-hard-to-manipulate}, considers constructive coalitional weighted manipulation.
This means that there is a group (or \emph{coalition}) of manipulators who jointly seek to make their preferred candidate win a weighted election.
In a \emph{weighted} election, voters may not all have the same voting power; rather, each voter is assigned a nonnegative integer weight, and if the voting system (as is the case for SkS) uses scores, these weights need to be taken into account when determining the candidates' scores.
For SkS, a weighted election is represented as $(C,V, W_V)$, where $W_V = (w_1, \dots, w_n)$ is a list of nonnegative integer weights with $w_i$ being the weight of the $i$-th voter.
The majority threshold now is
$\maj(V) = \floor*{\nicefrac{\left(\sum_{v \in V} w_v\right)}{2}} + 1$,
and $\scorei{i}(c)$ for $i \in [m]$ and $c\in C$ now is the sum of the weights of those voters who rank $c$ among their top-$i$ positions, i.e., $\scorei{i}(c) = \sum_{v\in V}\left(\eins_{\mathrm{pos}_v(c) \leq i} \cdot w_v\right)$.
The definition of $\sumposi{i}(c)$ is analogously weighted, i.e., $\sumposi{i}(c) = \sum_{v\in V}\left(\mathrm{pos}_v(c) \cdot \eins_{\mathrm{pos}_v(c) \leq i} \cdot w_v\right)$, and SkS for weighted elections works as described in Definition~\ref{def:SkS}.
Now, we are ready to define the most general manipulation problem for any voting system~$\mathcal{E}$:


\EP{CCWM}{$\mathcal{E}$-Constructive-Coalitional-Weighted-Manipulation ($\mathcal{E}$-CCWM)}
   {A set $C$ of candidates, a list $V$ of nonmanipulative votes over $C$ each having a nonnegative integer weight collected in a list~$W_V$, a list $W_S$ of the manipulators' nonnegative integer weights, where the manipulators' votes in $S$ are as yet unspecified, $V\cap S = \emptyset$, and a distinguished candidate $c\in C$.}
   {Can the votes in $S$ be set such that $c$ is the unique $\mathcal{E}$ winner of the weighted election $(C, V \cup S, W_{V \cup S})$?}

The following special variants of $\mathcal{E}$-\textsc{CCWM} have been studied in COMSOC (see, e.g., \cite{con-wal:b:handbook-comsoc-manipulation,bau-rot:b-2nd-edition:economics-and-computation-preference-aggregation-by-voting}):
$\mathcal{E}$-\textsc{Constructive-Coalitional-Manipulation} ($\mathcal{E}$-\textsc{CCM}) where the considered elections are \emph{unweighted} (i.e., every voter has unit weight) and
$\mathcal{E}$-\textsc{Constructive-Manipulation} ($\mathcal{E}$-\textsc{CM}) where the considered elections are \emph{unweighted} and there is only a \emph{single manipulator}.
Further, when the goal of the manipulation is not to make a distinguished candidate $c$ win alone but to prevent $c$ from becoming the sole winner, we obtain the \emph{destructive} variants of the above three problems:
$\mathcal{E}$-\textsc{Destructive-Coalitional-Weighted-Manipulation} ($\mathcal{E}$-\textsc{DCWM}),
$\mathcal{E}$-\textsc{Destructive-Coalitional-Mani\-pulation} ($\mathcal{E}$-\textsc{DCM}), and
$\mathcal{E}$-\textsc{Destructive-Manipulation} ($\mathcal{E}$-\textsc{DM}).

We focus on the \emph{unique-winner model}, as defined in the above problems.\footnote{One could also study the \emph{nonunique-winner model} where the constructive goal is merely to make $c$ a winner (possibly among several winners), and the destructive goal is to ensure that $c$ does not win at all.}
We expect the reader to be familiar with the notions of $\np$-hardness and $\np$-completeness based on the polynomial-time many-one reducibility $\textmanyone$ (see, e.g., \cite{gar-joh:b:int,kar:b:reducibilities,pap:b:complexity,rot:b:cryptocomplexity})

\begin{theorem}
  \label{thm:sks-ccwm-np-complete}
	SkS-\textsc{CCWM} is \np-complete, even for instances
	with only four candidates.
\end{theorem}

\np-hardness of SkS-\textsc{CCWM} can be shown by suitably extending the proof of Bucklin-\textsc{CCWM}
being $\np$-hard, which is due to Faliszeswki et al.~\cite{fal-rei-rot-sch:j:manipulation-bribery-campaign-management-in-bucklin-fallback-voting},
via a reduction from the well-known $\np$-complete problem \textsc{Partition}:
Given a list $(a_1, a_2,\dots, a_k)$ of positive integers with $\sum_{i=1}^k a_i = 2K$, does there exist a subset $A \subseteq [k]$ such that $\sum_{i \in A} a_i = \sum_{i \in [k] \setminus A} a_i = K$?
In particular, we have to take care of all cases in which the sum of positions plays a role in the winner determination.

The complexity of SkS-\textsc{CCM} and SkS-\textsc{CM} remains open.

We now turn to the destructive case and study the complexity of SkS-\textsc{DCWM}.
Unlike its constructive variant, this problem is easy to solve, which can be achieved by employing the algorithm of Faliszeswki et al.~\cite{fal-rei-rot-sch:j:manipulation-bribery-campaign-management-in-bucklin-fallback-voting} for Bucklin-\textsc{DCWM}
and by appropriately handling the cases in which the sum of positions plays a role for determining the SkS winners.
Since \textsc{DCM} is a special case of \textsc{DCWM} and \textsc{DM} in turn is a special case of \textsc{DCM}
(see, e.g., \cite{bau-rot:b-2nd-edition:economics-and-computation-preference-aggregation-by-voting}),
we obviously have:
$\textsc{DM} \manyone \textsc{DCM} \manyone \textsc{DCWM}$.
%

\begin{theorem}\label{thm:sks-dcwm-in-p}
SkS-\textsc{DCWM} (and hence, SkS-\textsc{DCM} and SkS-\textsc{DM}) are in~$\p$.
\end{theorem}


\section{Complexity of Control in SkS}
\label{sec:control-sks}

In the previous section, a coalition of manipulators seeked to influenced the election in their favor.
In a similar fashion, an external agent (called the \emph{chair}) having the power to change the structure of an election, may
exert this control power, again with the goal of making some favored candidate win (or preventing a despised candidate's victory).
The study of electoral control was initiated by Bartholdi et al.~\cite{bar-tov-tri:j:control} for the constructive case and by Hemaspaandra et al.~\cite{hem-hem-rot:j:destructive-control} for the destructive case.
Since then, control has been investigated for many voting systems (see, e.g.,
\cite{bau-rot:b-2nd-edition:economics-and-computation-preference-aggregation-by-voting,fal-rot:b:handbook-comsoc-control-and-bribery}), and in particular for Bucklin voting by Erd\'{e}lyi et al.~\cite{erd-fel-rot-sch:j:control-in-bucklin-and-fallback-voting-theoretical,erd-fel-rot-sch:j:control-in-bucklin-and-fallback-voting-experimental}.
Out of the many standard control types from the literature (again, see \cite{bau-rot:b-2nd-edition:economics-and-computation-preference-aggregation-by-voting,fal-rot:b:handbook-comsoc-control-and-bribery}),
we focus on a few select cases, namely on (constructive and destructive) control by either adding voters or deleting candidates, starting with the latter type of control for any voting system~$\mathcal{E}$:

\EP{CCDC}{$\mathcal{E}$-Constructive-Control-by-Deleting-Candidates ($\mathcal{E}$-CCDC)}
{A set $C$ of candidates, a list $V$ of votes over~$C$, a distinguished candidate $c\in C$, and a nonnegative integer $k$.}
{Does there exist a subset of candidates $C' \subseteq C$, $\abs{C'} \leq k$, such that $c$ is the unique $\mathcal{E}$ winner of the election $(C \setminus C', V)$, where the votes in $V$ are now restricted to the candidates in $C \setminus C'$?}

$\mathcal{E}$-\textsc{Destructive-Control-by-Deleting-Candidates} ($\mathcal{E}$-\textsc{DCDC}), the destructive variant of $\mathcal{E}$-\textsc{CCDC}, is defined analogously, except that now the chair's goal is to ensure that the distinguished candidate $c$ is \emph{not} a unique winner of $(C \setminus C', V)$ (i.e., $c$ may still win, but only along with other candidates), and it is not allowed to simply delete~$c$ (i.e., we require $c \not\in C'$).
In $\mathcal{E}$-\textsc{Constructive-Control-by-Adding-Voters} ($\mathcal{E}$-\textsc{CCAV}), in addition to the candidates~$C$, the votes~$V$ over~$C$, the distinguished candidate $c \in C$, and a nonnegative integer $k$ (now as an addition limit), we are given a list of additional votes~$U$ over~$C$ of as yet \emph{unregistered} voters, and we ask whether there exists a sublist of votes $U' \subseteq U$ with $\abs{U'} \leq k$ such that $c$ is the unique $\mathcal{E}$ winner of the election $(C, V \cup U')$.
Again, the destructive variant of $\mathcal{E}$-\textsc{CCAV}, $\mathcal{E}$-\textsc{Destructive-Control-by-Adding-Voters} ($\mathcal{E}$-\textsc{DCAV}), is defined analogously, except that the chair's goal now is to prevent the distinguished candidate $c$ from becoming a unique winner.
Again, we focus on the unique-winner model.

We now present our results on the above control problems for SkS.

\begin{theorem}\label{thm:sks-ccdc-dcdc-ccav}
	SkS-\textsc{CCDC}, SkS-\textsc{DCDC}, and SkS-\textsc{CCAV} are $\np$-complete in the unique-winner model.
\end{theorem}

For all of the above control problems we can either directly apply or adapt the proof of the respective control problem for Bucklin voting due to Erd\'{e}lyi et al.~\cite{erd-fel-rot-sch:j:control-in-bucklin-and-fallback-voting-theoretical}.

In contrast, SkS is vulnerable to \emph{destructive} control by adding voters.
We note that Bucklin-\textsc{DCAV} is in~$\p$ as well, i.e., Erd\'{e}lyi et al.~\cite{erd-fel-rot-sch:j:control-in-bucklin-and-fallback-voting-theoretical} provide a polynomial-time algorithm solving this problem.
We know that by Theorem~\ref{thm:properties-sks-bucklin}, any unique Bucklin winner is also a unique SkS winner and thus any Yes-instance of Bucklin-\textsc{DCAV} is also a Yes-instance of SkS-\textsc{DCAV} in the unique-winner model.
The reverse, however, is not always true, since in the case of multiple Bucklin winners, the despised candidate could still be a unique winner in SkS due to the sum of positions.
Therefore, in order to adapt it to SkS, we need to extend the existing algorithm for Bucklin-\textsc{DCAV} so as to take each additional case into consideration.

\begin{theorem}
  \label{thm:sks-dcav}
  In the unique-winner model, SkS-\textsc{DCAV} is in~$\p$.
\end{theorem}

Due to limited space, we only provide a sketch of the adaptions and the resulting algorithm.
Recall that by Theorem~\ref{thm:properties-sks-bucklin}, a unique Bucklin winner is also a unique SkS winner, and every SkS winner is also a Bucklin winner.
Additionally, Bucklin-\textsc{DCAV} is in $\p$ for both winner models, as has been shown by Erd\'{e}lyi et al.~\cite{erd-fel-rot-sch:j:control-in-bucklin-and-fallback-voting-theoretical}.
We can use these facts
by first running the Bucklin-\textsc{DCAV} algorithms before handling the cases specific to SkS in which the sum of positions comes into play.
Essentially, if
destructive control by adding voters is possible for Bucklin in the nonunique-winner model (i.e, if the despised candidate is not among the Bucklin winners after this control action), it immediately follows that destructive control by adding voters is possible for SkS
in the unique-winner model as well.
Moreover, if 
destructive control by adding voters for Bucklin is impossible in both winner models, it immediately follows that this control action is also impossible for SkS
in the unique-winner model.

We now consider the only remaining case where
destructive control by adding voters for Bucklin is impossible in the nonunique-winner model but possible in the unique-winner model.
That is, no candidate strictly beats the despised candidate, but it is possible to find a list of votes to add such that at least one candidate wins together with the despised candidate as a Bucklin co-winner.
For SkS, this is a necessary, yet not a sufficient condition for
destructive control by adding voters to be possible in the unique-winner model:
If two candidates tie in Bucklin score during some stage, since the despised candidate can still have a lower sum of positions, we cannot guarantee that the despised candidate is not a unique SkS winner.
Fortunately, it is possible---even though nontrivial---to check in polynomial time whether the current Bucklin co-winner can actually beat the despised candidate in SkS.

\section{Conclusion and Future Work}

\begin{table}
	\caption{Overview of complexity results for manipulation and control in SkS. R means resistance and V means vulnerability.}
	\label{tab:sks-complex-m-c-overview}
	\begin{tabularx}{\linewidth}{C|CCC}
		\toprule
		 & CWM & CDC & CAV \\
		\midrule
		Constructive (C) & R [Theorem~\ref{thm:sks-ccwm-np-complete}] &  R [Theorem~\ref{thm:sks-ccdc-dcdc-ccav}] &  R [Theorem~\ref{thm:sks-ccdc-dcdc-ccav}] \\
		Destructive (D) & V [Theorem~\ref{thm:sks-dcwm-in-p}] & R [Theorem~\ref{thm:sks-ccdc-dcdc-ccav}] &  V [Theorem~\ref{thm:sks-dcav}] \\
		\bottomrule
	\end{tabularx}
\end{table}

We have introduced the \emph{Skating System Single}, a voting system that is inspired by the Skating System which is used in ballroom tournaments, i.e., dance competitions.
We have studied SkS from an axiomatic point of view and have analyzed the computational complexity of manipulation for it; in addition, we have started the complexity-theoretic analysis of electoral control for SkS (see Tables~\ref{tab:properties} and \ref{tab:sks-complex-m-c-overview} for an overview of our results).
Since SkS has many similarities with Bucklin voting, it is not surprising that these two voting systems share many of their properties.
However, this is not always the case: For the property of positive responsiveness, we show that SkS satisfies it, whereas Bucklin does not.
While SkS certainly stands on its own, it can also be viewed as a refinement of Bucklin voting (as demonstrated by Theorem~\ref{thm:properties-sks-bucklin} and Lemma~\ref{lem:winner-sks-ties-restricted-no-sum-pos}) and can thus be applied in many general election settings aside from ballroom tournament finals.
Further, the proofs for SkS are technically more complicated due to the tie-breaking rules involving not only scores but also the candidates' sums of positions in the votes.

As to future research, it would be interesting to study further axiomatic properties of SkS (e.g., the \emph{Smith criterion}~\cite{smi:j:aggregation-of-preferences-with-variable-electorate}\footnote{Note that the \emph{strict Smith set} is also known as the \emph{Schwartz set}~\cite{sch:j:rational-policy-evaluation,sch:j:rationality-and-the-myth-of-the-maximum}.})
and to expand our study of attacking SkS to other manipulation and control scenarios as well as to bribery, which have been studied intensively for Bucklin voting~\cite{erd-fel-rot-sch:j:control-in-bucklin-and-fallback-voting-theoretical,erd-fel-rot-sch:j:control-in-bucklin-and-fallback-voting-experimental,fal-rei-rot-sch:j:manipulation-bribery-campaign-management-in-bucklin-fallback-voting}.
In particular, it would be highly interesting to detect further differences between the results on SkS and those on Bucklin voting.

\bigskip
\noindent
\textbf{Acknowledgements.}
This work was supported in part by Deutsche Forschungsgemeinschaft under DFG research grant RO\nobreakdash-1202/21\nobreakdash-2 (project 438204498).

\bibliographystyle{eptcs}
\bibliography{skating}



\end{document}